\makeatletter
\@namedef{ver@e-jc.sty}{2023/07/03 replaced by tbpreprint}
\@namedef{opt@e-jc.sty}{amsmath}
\@namedef{ver@lmodern.sty}{2009/10/30 replaced by libertinus}
\@namedef{ver@amssymb.sty}{2013/01/14 replaced by libertinust1math}
\let\tbkit@documentclass\documentclass
\renewcommand\documentclass[2][]{\tbkit@documentclass{tbpreprint}}
\AddToHook{class/tbpreprint/after}{%
  \def\copy@right{}%
  \def\the@dateline{}%
  \newcommand\Copyright[1]{\def\copy@right{\copyright\thinspace\ignorespaces#1}}%
  \newcommand\dateline[3]{}%
  \newcommand\MSC[1]{\msc{#1}}%
  \long\def\authortext#1#2{\tbkit@partir#2ORCID \tbkit@fin.\@nil}%
  \tbtipo{}%
  \tbnota{\the@dateline}%
  \RequirePackage{tbdiagram}%
}
\def\doi{\begingroup\catcode`\_=12\relax\doi@arg}
\def\doi@arg#1{\href{https://doi.org/#1}{\texttt{doi:#1}}\endgroup}

\def\tbkit@fin{\tbkit@fin}
\long\def\tbkit@partir#1ORCID #2.#3\@nil{%
  \affiliation{#1}%
  \def\tbkit@orcid{#2}%
  \ifx\tbkit@orcid\tbkit@fin\let\tbkit@orcid\@empty\fi}
\let\tbkit@orcid\@empty
\def\tbkit@primeralinea#1\\#2\@nil{#1}
\AtBeginDocument{%
  \ifx\tbkit@orcid\@empty\else\protected@edef\@author{\@author\noexpand\orcid{\tbkit@orcid}}\fi
  \shorttitle{\expandafter\tbkit@primeralinea\@title\\\@nil}%
}
\def\tbkit@sindospuntos#1:#2\@nil{#1}
\AddToHook{env/abstract/begin}{\renewcommand\textbf[1]{\tbetiqueta{\tbkit@sindospuntos#1:\@nil}\quad
  \ignorespaces}\def\bigskip{\vspace{1.2ex}}}
\def\tbkit@antes#1#2{\global\@namedef{tbkit@antes@#1}{#2}}
\def\tbkit@dispara#1{\@ifundefined{tbkit@antes@#1}{}{%
  \@nameuse{tbkit@antes@#1}\global\expandafter\let\csname tbkit@antes@#1\endcsname\relax}}
\AddToHook{cmd/section/before}{\tbkit@dispara{s\the\numexpr\value{section}+1\relax}}
\makeatother

\makeatletter
\tbkit@antes{s3}{% The link decomposition of one cell as a chain of statements, for the note. Adapted from
\begin{figure}[tbp]
\centering
\begin{tikzpicture}[tb diagrama, x=1mm, y=1mm,
    caja/.style={text width=26mm, font=\sffamily\small}, fila/.style={minimum height=22mm}]
  \node[tb datos, caja, fila] (cov) at (0,0)
    {a $(v,k,t)$-covering with $L$ blocks\\[1pt] {\footnotesize (assumed, to be refuted)}};
  \node[tb proceso, caja, fila] (lem) at (36.5,0)
    {step (i) UNSAT\\[1pt] {\footnotesize no $(v{-}1,k{-}1,t{-}1)$-\allowbreak covering with $r-1$ blocks, so every degree is $\ge r$}};
  \node[tb neutro, caja, fila] (deg) at (73,0)
    {$(r+1)v>kL$\\[1pt] {\footnotesize so some point $p$ has degree exactly $r$}};
  \node[tb datos, caja, fila] (link) at (109.5,0)
    {the link of $p$\\[1pt] {\footnotesize a $(v{-}1,k{-}1,t{-}1)$-\allowbreak covering with $r$ blocks}};
  \draw[tb flecha] (cov) to (lem);
  \draw[tb flecha] (lem) to (deg);
  \draw[tb flecha] (deg) to (link);

  \node[tb datos, caja, fila] (list) at (0,-40)
    {class list $c_1,\dots,c_N$\\[1pt] {\footnotesize from any source, $N$ from \clasesMinNota{} to \clasesMaxNota{} per cell}};
  \node[tb proceso, caja, fila] (flist) at (36.5,-40)
    {$\Flista$ UNSAT\\[1pt] {\footnotesize the link is isomorphic to $M_{c_i}$ for some $i$}};
  \begin{scope}[ext/.style={tb proceso, text width=22mm, minimum height=6mm, inner ysep=2pt,
      font=\sffamily\footnotesize}]
    \node[ext] (e1) at (73,-29) {$\Fclase(c_1)$ UNSAT};
    \node[ext] (e2) at (73,-37.5) {$\Fclase(c_2)$ UNSAT};
    \node[font=\sffamily\small, text=tbgristexto] (edots) at (73,-44.5) {$\vdots$};
    \node[ext] (eN) at (73,-52) {$\Fclase(c_N)$ UNSAT};
  \end{scope}
  \node[tb terminal, text width=24mm] (fin) at (109.5,-40)
    {no covering with $L$ blocks\\[1pt] $C(v,k,t)\ge L+1$};
  \draw[tb flecha] (list) to (flist);
  \draw[tb flecha] (link.south) to ++(0,-8.5) -| ([xshift=4mm]flist.north);
  \draw[tb flecha] (flist.east) to (e1.west);
  \draw[tb flecha] (flist.east) to (e2.west);
  \draw[tb flecha] (flist.east) to (eN.west);
  \draw[tb flecha] (e1.east) to (fin.west);
  \draw[tb flecha] (e2.east) to (fin.west);
  \draw[tb flecha] (eN.east) to (fin.west);
  \node[tb anotacion, align=center, text width=34mm, below=1.5mm of eN]
    {one formula per class, or per group of up to \tamGrupo{} classes};
  \begin{scope}[on background layer]
    \node[tb grupo=for every class, fit=(e1)(eN), inner sep=3pt] {};
  \end{scope}
\end{tikzpicture}
\caption{The link decomposition of one cell as a chain of statements
(Theorem~\ref{thm:reduction}). Slate boxes are combinatorial objects, light amber boxes are formulas
that are refuted with proofs checked by cake\_lpr, the white box is arithmetic and the solid amber
box is the conclusion. The list of classes may come from any source: a missing class would make
$\Flista$ satisfiable.}\label{fig:kit-decomposition}
\end{figure}
}
\makeatother

\documentclass[12pt]{article}
\usepackage[amsmath]{e-jc}

\usepackage[T1]{fontenc}
\usepackage{lmodern}
\usepackage{amssymb}
\usepackage{booktabs}
\usepackage{array}
\usepackage{xcolor}
\usepackage{microtype}
\usepackage{xurl}

\newcommand{\lex}{\preceq}
\newcommand{\Flista}{\Phi_{\mathrm{list}}}
\newcommand{\Fclase}{\Phi_{\mathrm{ext}}}
\newcommand{\sort}{\operatorname{sort}}
\newcommand{\Img}{\operatorname{Im}}

\newcommand{\zdoi}[1]{\href{https://doi.org/#1}{doi:#1}}
\newcommand{\nHuecoUno}{443}
\newcommand{\nAplica}{433}
\newcommand{\nNoExacto}{7}
\newcommand{\nFalla}{3}

\newcommand{\nTdos}{282}

\newcommand{\rLRseisMin}{0.46}
\newcommand{\rLRseisMax}{0.60}

\newcommand{\nTtresRcinco}{12}

\newcommand{\tamGrupo}{100}

\newcommand{\nAgrupadas}{6}
\newcommand{\vAgrupadasMin}{30}

\newcommand{\depositoGib}{143}

\newcommand{\nCeldas}{25}
\newcommand{\nCeldasPalabra}{twenty-five}
\newcommand{\nCeldasLarga}{16}
\newcommand{\nCeldasNuevas}{9}
\newcommand{\valoresTexto}{9, 10 and 11}
\newcommand{\tConjunto}{$t\in\{3,4,5\}$}
\newcommand{\fechaRepositorio}{20 October 2026}
\newcommand{\nSchonheimIgual}{15}
\newcommand{\nSchonheimMas}{10}
\newcommand{\fraseHorsleySingh}{none of the twenty-five cells appears in Tables~1 and~2}
\newcommand{\rLMinNota}{0.50}
\newcommand{\libresMaxNota}{4}
\newcommand{\clasesMinNota}{1}
\newcommand{\clasesMaxNota}{9{,}375}

\newcommand{\nFormulasCod}{2{,}370}
\newcommand{\nFormulasTotal}{4{,}740}

\newcommand{\nListasControl}{20}

\newcommand{\gibNuevas}{5.4}
\newcommand{\nRejilla}{1{,}621}
\newcommand{\frasePendientes}{ We also attempted seven cells with $t=3$, $r\in\{6,7\}$ and 2 to 38 classes. In each of them an extension formula produced a DRAT proof of more than 3.5\,GiB before the solver finished, more than drat-trim can check in the memory we gave it, and we stopped.}
\newcommand{\nRseisNota}{32}
\newcommand{\nRseisEnum}{21}
\newcommand{\rseisClasesMin}{1}
\newcommand{\rseisClasesMax}{1{,}002{,}213}
\newcommand{\nRseisTabla}{6}

\dateline{Oct 21, 2026}{TBD}{TBD}
\ifdefined\ejcmode\else
  \makeatletter
  \renewcommand{\ps@plain}{\renewcommand{\@oddfoot}{\hfil\thepage\hfil}\renewcommand{\@evenfoot}{\hfil\thepage\hfil}}
  \AtBeginDocument{\def\the@dateline{Preprint, 21 October 2026\\[0.5ex]\copy@right}}
  \makeatother
\fi
\MSC{05B40, 05B30, 68V05, 68V15}
\Copyright{The author. Released under the CC BY licence (International 4.0).}

\title{Links of minimum degree\\ and exact covering numbers}

\author{Alfonso Gippini Requeijo}
\authortext{}{Team Banzai S.L.U., Madrid, Spain (\email{alfonso@team-banzai.com}). ORCID 0009-0004-2109-0583.}

\begin{document}

\maketitle

\begin{abstract}
The covering number $C(v,k,t)$ is the least number of $k$-subsets, called blocks, of a $v$-set
such that every $t$-subset lies in one of them. We reduce the question whether a covering with $L$
blocks exists to the classification of links. If every point of such a covering lies in at least
$r$ blocks and $(r+1)v>kL$, some point lies in exactly $r$ blocks, and its link, the blocks through
it with the point removed, is a $(v-1,k-1,t-1)$-covering with $r$ blocks. It then suffices to list
these links up to isomorphism, to prove the list complete, and to show that no link in the list
extends to a covering with $L$ blocks. Completeness is certified by one unsatisfiable formula that
excludes the listed classes, so the program that lists the links need not be trusted. In general
the last two steps are decided by propositional satisfiability (SAT) solving, with every formula
written by two separately implemented encoders and every refutation checked by a formally verified
proof checker. For links with four blocks and $t=3$ we give the classification by hand, under a
condition on the parameters. As an application we determine \nCeldasPalabra{} covering numbers with
\tConjunto{} in which the published lower bound was one less than the size of a known covering,
among them $C(39,24,3)=9$.

\bigskip\noindent \textbf{Keywords:} covering designs, covering numbers, classification of
designs, SAT solving, verified proof checking
\end{abstract}

% ---------------------------------------------------------------------------------------------
\section{Introduction}\label{sec:intro}

A $(v,k,t)$-covering is a family $\mathcal{B}$ of $k$-subsets of a $v$-set $V$, called blocks,
such that every $t$-subset of $V$ is contained in at least one block. Blocks may repeat. The
covering number $C(v,k,t)$ is the least size of a $(v,k,t)$-covering. The classical lower bound is
due to Sch\"onheim~\cite{schonheim1964}:
\begin{equation}\label{eq:schonheim}
  C(v,k,t) \;\ge\; \left\lceil \frac{v}{k}\, C(v-1,k-1,t-1) \right\rceil .
\end{equation}
It counts incidences: the blocks through a point, with that point removed, form a
$(v-1,k-1,t-1)$-covering. Exact values and the best known bounds are collected in the La Jolla
Covering Repository of Gordon, Kuperberg and Patashnik~\cite{gkp1995}, now frozen as a dataset on
Zenodo~\cite{ljcr2026}, and in its successor, the Covering Repository~\cite{coveringrepository}. We
call a parameter set $(v,k,t)$ with its bounds a \emph{cell}.

\paragraph{The method.} Let $r$ be the least size of a $(v-1,k-1,t-1)$-covering. Every point of a
$(v,k,t)$-covering with $L$ blocks lies in at least $r$ blocks. The degrees of the $v$ points add
up to $kL$, so if $(r+1)v>kL$ some point $p$ lies in exactly $r$ blocks. Its link, the blocks
through $p$ with $p$ removed, is then a $(v-1,k-1,t-1)$-covering with exactly $r$ blocks, a link
of minimum degree. It fixes $r$ of the $L$ columns of the incidence matrix, and the question becomes whether one of
finitely many links, listed up to isomorphism, extends by $L-r$ further blocks.
Section~\ref{sec:method} makes this precise, including how to prove a list of links complete
without trusting the program that produced it. Section~\ref{sec:example} classifies by hand the
links with four blocks for $t=3$ under a condition on the parameters. The condition $(r+1)v>kL$ is
weaker than $kL=vr$, under which every point has degree $r$, and holds whenever $L$ is the
Sch\"onheim value (Section~\ref{sec:method}). Compared with the closest computer proofs (see
Related work), the method uses this weaker condition and certifies separately that the list of
links is complete.

\paragraph{The cells.} As of \fechaRepositorio, the date of every statement below about the two
repositories, the Covering Repository listed each of the \nCeldas{} cells of Table~\ref{tab:main}
with a covering of $L+1$ blocks and a lower bound of $L$.

\begin{theorem}\label{thm:main}
For each cell $(v,k,t)$ of Table~\ref{tab:main}, with $L$ as given there, $C(v,k,t)=L+1$.
\end{theorem}

The covering numbers obtained are \valoresTexto. The upper bounds are the known coverings, which
we read from the La Jolla dataset and checked by a program that tests every $t$-subset. The
content of the theorem is that no covering with $L$ blocks exists. In \nSchonheimIgual{} of the
cells the recorded lower bound is the Sch\"onheim bound~\eqref{eq:schonheim}, and in the other
\nSchonheimMas{} it is one more. In all of them the link fixes a fraction $r/L\ge\rLMinNota$ of the
columns and $L-r\le\libresMaxNota$, and the classification by hand of Section~\ref{sec:example}
covers three of them, among them $C(39,24,3)$, the example of that section. The
bounds of Todorov~\cite{todorov1989} concern $t=2$, and \fraseHorsleySingh{} of Horsley and
Singh~\cite{horsleysingh2018}, which list the $(v,k)$ with $k\le 40$ and $t\in\{3,4,5\}$ where
their bounds improve on the best bound known to them. To our knowledge none of these values had
been established before, except $C(18,11,3)$ and $C(20,12,3)$, which we settled in an earlier data
deposit~\cite{zenodo2cells} by a split on degree sequences and which the repositories do not yet
record.

\paragraph{Related work.} Two recent computer-assisted proofs of $C(12,6,4)=41$ also fix a link.
In Krug's proof~\cite{krug2026}, $kL=vr$ with $L=40$ and $r=20$, so every point of a hypothetical
$40$-block covering has degree exactly $20$, every link is an optimal $(11,5,3)$-covering, and
this covering is unique up to isomorphism (Krug's Section~2 and Proposition~15). The proof package
of Bertram~\cite{bertram2026} fixes the nine blocks through a pair of points and enumerates the
orbits of such links. In our cells $vr<kL<v(r+1)$, so at least $(r+1)v-kL$ points, but not all,
have degree $r$, and there are up to $\clasesMaxNota$ classes of links. What has to be added is a way to make
the classification of links trustworthy. Blocking every model found and ending with an
unsatisfiable formula is the usual way to enumerate all models with a SAT
solver~\cite{mcmillan2002,todasoh2016}. We use it once, to certify a list produced outside the
solver.

% ---------------------------------------------------------------------------------------------
\section{Preliminaries}\label{sec:prelim}

Throughout, $2\le t\le k<v$. A covering with $L$ blocks is given by its incidence matrix
$Y\in\{0,1\}^{v\times L}$, with $Y_{p,j}=1$ when point $p$ lies in block $j$. So $Y$ is the
incidence matrix of a $(v,k,t)$-covering if and only if every column has exactly $k$ ones and
every $t$ rows have a common column of ones. Two matrices are isomorphic if one is obtained from
the other by permuting rows and permuting columns. The degree $d_p$ of a point is the number of
ones in row $p$.

For binary vectors $a,b$ of the same length, $a\lex b$ means that $a=b$ or $a_i<b_i$ at the first
index $i$ where they differ. Matrices of the same shape are compared by $\lex$ on their row
readings, the concatenation of their rows, row $0$ first.

\begin{lemma}\label{lem:monotone}
If a $(v,k,t)$-covering with $L'\le L$ blocks exists, one with exactly $L$ blocks exists.
\end{lemma}

\begin{proof}
Repeat one block $L-L'$ times.
\end{proof}

\begin{lemma}\label{lem:link}
Let $Y$ be a $(v,k,t)$-covering and $p$ a point. Then $d_p\ge 1$, and deleting row $p$ and the
columns of the blocks that miss $p$ leaves the incidence matrix of a $(v-1,k-1,t-1)$-covering with
$d_p$ blocks, the \emph{link} of $p$.
\end{lemma}

\begin{proof}
A $t$-set containing $p$ exists because $v\ge t$, and it lies in a block, so $d_p\ge 1$. A
$(t-1)$-set $S$ of points other than $p$ gives the $t$-set $S\cup\{p\}$, which lies in some block
$B$, and then $p\in B$ and $S\subseteq B\setminus\{p\}$.
\end{proof}

\begin{lemma}[Degree lemma]\label{lem:degree}
If there is no $(v-1,k-1,t-1)$-covering with $r-1$ blocks, every point of every
$(v,k,t)$-covering has degree at least $r$.
\end{lemma}

\begin{proof}
If $d_p\le r-1$, the link of $p$ is a $(v-1,k-1,t-1)$-covering with $d_p$ blocks, and
Lemma~\ref{lem:monotone} turns it into one with $r-1$ blocks.
\end{proof}

With $r=C(v-1,k-1,t-1)$ and $\sum_p d_p=kL$, Lemma~\ref{lem:degree} gives~\eqref{eq:schonheim}.
The next lemma is the standard fact that row and column symmetries can be broken by lexicographic
order~\cite{doublelex2002}, in the form in which we use it.

\begin{lemma}[Lex-least representative]\label{lem:lexleast}
Let $X$ be a binary matrix, let its rows be partitioned into parts of consecutive indices, and let
$J$ be a set of consecutive column indices. Let $G$ be the group that permutes the rows inside each
part and permutes the columns in $J$, and let $X^*$ have the least row reading in the orbit $GX$.
Then \textup{(i)} if rows $q$ and $q+1$ lie in the same part, row $q$ of $X^*$ $\lex$ row
$q+1$, and \textup{(ii)} if $j,j+1\in J$, column $j$ of $X^*$ $\lex$ column $j+1$.
\end{lemma}

\begin{proof}
(i) Otherwise exchanging rows $q$ and $q+1$ leaves rows $0,\dots,q-1$ unchanged and makes row $q$
smaller. (ii) Otherwise let $i$ be the first row where columns $j$ and $j+1$ differ, so that
$X^*_{i,j}=1$ and $X^*_{i,j+1}=0$. Exchanging the two columns leaves rows $0,\dots,i-1$ unchanged,
and in row $i$ it keeps positions $0,\dots,j-1$ and turns position $j$ into $0$. In both cases
an element of $G$ decreases the row reading.
\end{proof}

With one part and $J$ equal to all columns, Lemma~\ref{lem:lexleast} shows that every binary
matrix is isomorphic to a \emph{doubly lex} matrix, one whose rows and columns are both
$\lex$-nondecreasing.

% ---------------------------------------------------------------------------------------------
\section{The link decomposition}\label{sec:method}

\begin{proposition}\label{prop:mindeg}
Let $Y$ be a $(v,k,t)$-covering with $L$ blocks in which every point has degree at least $r$. If
$(r+1)v>kL$, some point has degree exactly $r$.
\end{proposition}

\begin{proof}
If every degree were at least $r+1$, then $kL=\sum_p d_p\ge (r+1)v$.
\end{proof}

The hypothesis holds whenever $L$ is the Sch\"onheim value: if $\rho=C(v-1,k-1,t-1)$ and
$L=\lceil v\rho/k\rceil$, then $kL<v\rho+k\le v(\rho+1)$. So the reduction below with $r=\rho$
applies to every cell whose question is whether the Sch\"onheim bound is attained, and to some
cells above it, such as the \nSchonheimMas{} cells of Table~\ref{tab:main} in which $L$ is one more.

\paragraph{Links and the completeness formula.} Fix $n=v-1$, $m=k-1$, $s=t-1$ and $b=r$. The
\emph{type} of a row of an $n\times b$ matrix $Z$ is the set $R\subseteq\{0,\dots,b-1\}$ of
columns where it has a $1$. Up to row permutation $Z$ is determined by its type counts
$(n_R)_R$, and $Z$ is the incidence matrix of an $(n,m,s)$-covering with $b$ blocks if and only if
$\sum_R n_R=n$, every column $j$ satisfies $\sum_{R\ni j}n_R=m$, and any at most $s$ of the types
that occur have a common element (here $n\ge s$ is used to complete $j\le s$ types to $s$ rows). An
isomorphism class of such coverings is an orbit of count vectors under $S_b$ acting on the columns.

A \emph{class list} is a finite sequence $c_1,\dots,c_N$ of count vectors. For a count vector $c$
let $M_c$ be the $n\times b$ matrix with rows of the types given by $c$, sorted
$\lex$-nondecreasingly. For $\sigma\in S_b$ let $M_c\sigma$ be the matrix whose column $\sigma(j)$
is column $j$ of $M_c$, and let $\sort(X)$ sort the rows of $X$ $\lex$-nondecreasingly. Let
$\Img^\circ(c)$ be the set of matrices $\sort(M_c\sigma)$, $\sigma\in S_b$, whose columns are
$\lex$-nondecreasing.

\begin{definition}[Completeness formula]\label{def:list}
$\Flista(n,m,s,b\mid c_1,\dots,c_N)$ states that there is $Z\in\{0,1\}^{n\times b}$ such that
(1) every column has exactly $m$ ones, (2) every $s$ rows have a common column of ones, (3) the
rows are $\lex$-nondecreasing, (4) the columns are $\lex$-nondecreasing, and (5) $Z\notin
\Img^\circ(c_i)$ for $i=1,\dots,N$.
\end{definition}

\begin{theorem}\label{thm:list}
If $\Flista(n,m,s,b\mid c_1,\dots,c_N)$ is unsatisfiable, every $(n,m,s)$-covering with $b$ blocks
is isomorphic to some $M_{c_i}$.
\end{theorem}

\begin{proof}
By Lemma~\ref{lem:lexleast} a given covering is isomorphic to a doubly lex matrix $D$. This $D$
satisfies (1) to (4), so it violates (5): $D=\sort(M_{c_i}\sigma)$ for some $i$ and $\sigma$, and
$D$ is isomorphic to $M_{c_i}$.
\end{proof}

Nothing is assumed about how the list was produced: a missing class leaves the formula
satisfiable. Only images with ordered columns need to be blocked, at most $b!$ per class, because
(4) excludes the others. If in addition the $M_{c_i}$ are coverings and pairwise non-isomorphic, $N$
is the number of isomorphism classes of $(n,m,s)$-coverings with $b$ blocks. For the \nCeldas{} link
parameter sets that occur here we checked both conditions with a program that uses neither the
enumerator nor a solver, which gives the column ``classes'' of Table~\ref{tab:main}. The
covering numbers do not depend on that check.

\paragraph{Producing the lists.} Since Theorem~\ref{thm:list} accepts a list from any source, the
proof does not depend on the correctness of the enumerator, which matters only for speed. We
generate count vectors up to the action of $S_b$ by canonical augmentation~\cite{mckay1998}, adding
one block at a time: a child is kept only if the new block lies in the orbit, under the
automorphism group of the child, of the block that a canonical deletion removes, and if no sibling
already kept has the same canonical form. Partial objects are pruned by necessary conditions for
completion, such as the Sch\"onheim bound for the derived coverings at a point and at a pair of
points. Our first
program, a search over labelled supports, did not finish the links with six or seven blocks within
the time we gave it, and does not accept $s\ge 4$. Canonical augmentation lists both kinds and gives
the lists of the last \nCeldasNuevas{} cells of Table~\ref{tab:main}. The two programs give the
same lists on the first \nCeldasLarga{} cells and on $\nRejilla$ small parameter sets, and the new one
agrees with a brute force search on small instances. Both programs and these tests are in
version~2.0 of the data deposit.

\paragraph{The extension formula.} Let $c$ be a class and $f_1\lex\dots\lex f_{v-1}$ the rows
of $M_c$.

\begin{definition}[Extension formula]\label{def:ext}
$\Fclase(v,k,t,L,r\mid c)$ states that there is $X\in\{0,1\}^{v\times L}$ such that
(a) row $0$ has ones exactly in columns $0,\dots,r-1$,
(b) for $q=1,\dots,v-1$ the restriction of row $q$ to columns $0,\dots,r-1$ is $f_q$,
(c) every column has exactly $k$ ones,
(d) every $t$ rows have a common column of ones,
(e) every row has at least $r$ ones,
(f) for $1\le q\le v-2$ with $f_q=f_{q+1}$, row $q$ $\lex$ row $q+1$, and
(g) for $r\le j\le L-2$, column $j$ $\lex$ column $j+1$.
\end{definition}

\begin{theorem}[Link decomposition]\label{thm:reduction}
Let $2\le t\le k<v$, $L\ge r\ge 2$, and let $c_1,\dots,c_N$ be a class list for
$(v-1,k-1,t-1,r)$. Suppose that
\textup{(i)} there is no $(v-1,k-1,t-1)$-covering with $r-1$ blocks,
\textup{(ii)} $(r+1)v>kL$,
\textup{(iii)} every $(v-1,k-1,t-1)$-covering with $r$ blocks is isomorphic to some $M_{c_i}$, and
\textup{(iv)} $\Fclase(v,k,t,L,r\mid c_i)$ is unsatisfiable for $i=1,\dots,N$.
Then there is no $(v,k,t)$-covering with $L$ or fewer blocks.
\end{theorem}

\begin{proof}
By Lemma~\ref{lem:monotone} it is enough to exclude exactly $L$ blocks. Let $Y$ be a
$(v,k,t)$-covering with $L$ blocks. By (i) and Lemma~\ref{lem:degree} every degree is at least $r$,
and by (ii) and Proposition~\ref{prop:mindeg} some point $p$ has degree $r$. Renumber so that
$p=0$ and the blocks through $p$ are $0,\dots,r-1$. By Lemma~\ref{lem:link}, rows $1,\dots,v-1$
and columns $0,\dots,r-1$ form a $(v-1,k-1,t-1)$-covering with $r$ blocks, which by (iii) is isomorphic to $M_c$ for some $c=c_i$. Permute columns $0,\dots,r-1$ and
then rows $1,\dots,v-1$ so that the restriction of row $q$ to these columns is $f_q$. Row $0$ is
unaffected. The resulting $X$ satisfies (a) and (b) by construction, and (c), (d), (e) because
they are invariant under permutations of rows and columns.

Partition rows $1,\dots,v-1$ into maximal runs with equal $f_q$, put row $0$ in a part of its own,
and let $J=\{r,\dots,L-1\}$. Every element of the group $G$ of Lemma~\ref{lem:lexleast} preserves
(a) to (e): it fixes row $0$, which is zero on $J$, it permutes rows only among rows with the same
restriction to columns $0,\dots,r-1$, it does not move those columns, and (c), (d), (e) are
invariant. So the lex-least matrix of $GX$ satisfies (a) to (e), and by
Lemma~\ref{lem:lexleast} also (f) and (g). This contradicts (iv).
\end{proof}

Condition (iii) is established either by hand, as in Section~\ref{sec:example}, or by
Theorem~\ref{thm:list}, through the unsatisfiability of $\Flista$.

When the list is long, the formulas of (iv) can be merged. For a set $I$ of classes, let
$\Fclase(I)$ state that there are $X$ and selectors $(\varsigma_c)_{c\in I}$, exactly one of them
true, such that $X$ satisfies (a), (c), (d), (e), (g), and (b) and (f) for every $c$ whose
selector is true. A solution $X$ of $\Fclase(c)$, with $\varsigma_c$ true and the other selectors
false, solves $\Fclase(I)$. Conversely, in a solution of $\Fclase(I)$ exactly one selector
$\varsigma_c$ is true, and $X$ then satisfies (a) to (g) for $c$. So $\Fclase(I)$ is
unsatisfiable exactly when every $\Fclase(c)$, $c\in I$, is. That the propositional encodings of
$\Fclase(I)$ say this, that is, that an unselected class imposes nothing on $X$, is part of what has
to be trusted (Section~\ref{sec:trust}), and is shown clause by clause for both encodings, the same for all the cells, in the
long version~\cite{longversion}.

% ---------------------------------------------------------------------------------------------
\section{Links with four blocks, and \texorpdfstring{$C(39,24,3)=9$}{C(39,24,3)=9}}\label{sec:example}

In three of the cells the link has $r=4$ blocks and $t-1=2$, and steps (i) and (iii) can
be done by hand. For a set $R$ of blocks, $n_R$ is the number of points of type $R$, and $p_3(x)$
is the number of partitions of the integer $x\ge 0$ into at most three parts, with $p_3(0)=1$.

\begin{lemma}\label{lem:four}
Let $2n>3m$ and $d=5m-3n$. Then
\textup{(a)} there is no $(n,m,2)$-covering with $3$ blocks, and
\textup{(b)} an $(n,m,2)$-covering with $4$ blocks exists if and only if $d\ge 0$, and the number
of its isomorphism classes is then $\sum_{a}p_3(d-2a)$, over the integers $a$ with $0\le 2a\le d$.
\end{lemma}

\begin{proof}
Note that $m<2n/3<n$. A point in no block leaves the pairs through it uncovered, and if a point
lay only in block $j$, every other point would have to lie in block $j$, which has $m<n$ points. So every point lies in at least two
blocks, and the types that occur pairwise intersect. This gives (a), since three blocks have
$3m<2n$ incidences.

(b) Let $a_i$ be the number of points in exactly $i$ blocks. Counting points and incidences,
$a_2+a_3+a_4=n$ and $2a_2+3a_3+4a_4=4m$, so
\[
  a_3+2a_4=4m-2n, \qquad a_2=3n-4m+a_4 .
\]
The types of size two pairwise intersect, so either they all contain a common block or they are
the three pairs inside a set of three blocks: if two of them are $\{0,1\}$ and $\{0,2\}$, a third
avoiding $0$ must be $\{1,2\}$, and then every further one meets all three only if it is one of
them. In the second case the fourth block contains only
points of type size at least three, at most $a_3+a_4\le 4m-2n<m$ of them. So, after relabelling,
every type of size two contains block $0$. If $\{0,j\}$ did not occur for some $j\in\{1,2,3\}$,
block $j$ would again contain at most $4m-2n<m$ points. So $\{0,1\}$, $\{0,2\}$ and $\{0,3\}$ all
occur. Block $0$ is the only block common to them, so every automorphism and every isomorphism
between two such coverings maps block $0$ to block $0$. A type of size three either contains $0$ or
is $\{1,2,3\}$. Let $a=a_4$ and let $x$ be the number of points whose type has size three and
contains $0$. Block $0$ contains exactly the points of type size two or four and these $x$ points,
so
\[
  m=a_2+a+x=3n-4m+2a+x, \qquad\text{that is,}\qquad 2a+x=d .
\]
Hence $d\ge 0$. Let $x_j$ be the number of these $x$ points whose type contains $j\in\{1,2,3\}$, so that
$x_1+x_2+x_3=2x$.
Then $n_{\{1,2,3\}}=a_3-x=4m-2n-2a-x=n-m$, and the sum of column $j$ gives
$n_{\{0,j\}}=m-(n-m)-a-x_j=2m-n-a-x_j$. So the covering is determined by $a$ and by how the $x$
points are distributed over the types $\{0,1,2\}$, $\{0,1,3\}$, $\{0,2,3\}$.

Conversely, take $0\le a\le d/2$, put $x=d-2a$, distribute $x$ points over these three types in
any way, and define the other counts by the formulas above. Then $n_{\{0,j\}}\ge 2m-n-a-x=2n-3m+a>0$,
column $0$ has $\sum_j n_{\{0,j\}}+a+x=6m-3n-2a-x=m$ points, columns $1,2,3$ have $m$ points by
construction, there are $\sum_j n_{\{0,j\}}+x+a+(n-m)=n$ points, and the types used pairwise
intersect. Since isomorphisms fix block $0$ and induce any permutation of $\{1,2,3\}$, which
permutes the three types $\{0,i,j\}$ arbitrarily, the classes correspond to the pairs formed by
$a$ and a partition of $x$ into at most three parts.
\end{proof}

For $C(39,24,3)$ the link parameters are $(n,m)=(38,23)$, so $2n=76>69=3m$ and $d=1$. Hence
$r=C(38,23,2)=4$, and there is exactly one link class, with $a=0$, $x=1$:
\[
  n_{\{0,1\}}=7,\quad n_{\{0,2\}}=7,\quad n_{\{0,1,2\}}=1,\quad n_{\{0,3\}}=8,\quad n_{\{1,2,3\}}=15 .
\]
In the same way the link of $C(41,25,3)$, with $(n,m)=(40,24)$ and $d=0$, is unique, with
$n_{\{0,j\}}=8$ for $j=1,2,3$ and $n_{\{1,2,3\}}=16$, and the links of $C(37,23,3)$, with
$(n,m)=(36,22)$ and $d=2$, form $p_3(2)+p_3(0)=3$ classes: $a=0$ with the two points of type size
three containing $0$ of equal or of different types, and $a=1$ with $x=0$. In all three classes
$n_{\{1,2,3\}}=14$ and $n_{\{0,j\}}=8-a-x_j$. These agree with the classes column of
Table~\ref{tab:main} and with the certified lists. In all three cells $(r+1)v>kL$ (step (ii)), for
$C(39,24,3)$ because $5\cdot 39=195>192=24\cdot 8$. More generally, for $t=3$ and link
parameters with $3n/5\le m<2n/3$, Lemma~\ref{lem:four} gives $r=4$ and the complete list of links,
and, provided (ii) holds, only the extension step is left to decide.

\paragraph{The extension for $C(39,24,3)$.} The first four columns of the $39\times 8$ matrix are
now fixed, and the question is whether four further blocks of $24$ points cover every triple that
the link leaves uncovered. A triple of points other than $p$ is covered by a block through $p$
exactly when its three types share a block. So, for instance, every triple formed by a point of
type $\{0,1\}$, one of type $\{0,2\}$ and one of type $\{1,2,3\}$ must lie in one of the four free
blocks, and by (e) every point of type size two lies in at least two of them. We do not have a hand
proof that this is impossible. The extension formula is refuted by a SAT solver, like the formulas
of the other cells.

% ---------------------------------------------------------------------------------------------
\section{Application: \nCeldasPalabra{} cells}\label{sec:results}

Table~\ref{tab:main} lists the cells. In each, $r=C(v-1,k-1,t-1)$: step (i) excludes $r-1$
blocks, and the class lists are nonempty and consist of coverings with $r$ blocks.

\begin{table}[!htbp]
\centering\footnotesize
\caption{The cells. $L$ is refuted and $L+1$ is the size of the known covering, whose
construction is given as listed by the Covering Repository~\cite{coveringrepository}, with its
source or creator and its year. ``Classes'' is the number of isomorphism classes of
$(v-1,k-1,t-1)$-coverings with $r$ blocks.}\label{tab:main}
\vspace{4pt}
\setlength{\tabcolsep}{4pt}
\begin{tabular}{lrrlrrrl}
\toprule
cell & $L$ & $L+1$ & link & $r$ & $(r+1)v>kL$ & classes & construction of $L+1$ \\
\midrule
% generated by cifras_nota.py: the sixteen rows of tablas/resultados.tex, then the cells of cert-orb-2026-10-20
$C(18,11,3)$ & 9 & 10 & $(17,10,2)$ & 5 & $108 > 99$ & 92 & JCD article, 1996 \\
$C(20,12,3)$ & 9 & 10 & $(19,11,2)$ & 5 & $120 > 108$ & 49 & JCD article, 1996 \\
$C(23,14,3)$ & 9 & 10 & $(22,13,2)$ & 5 & $138 > 126$ & 358 & JCD article, 1996 \\
$C(25,15,3)$ & 9 & 10 & $(24,14,2)$ & 5 & $150 > 135$ & 215 & JCD article, 1996 \\
$C(27,16,3)$ & 9 & 10 & $(26,15,2)$ & 5 & $162 > 144$ & 101 & Kamal Fadlaoui, 2008 \\
$C(28,17,3)$ & 9 & 10 & $(27,16,2)$ & 5 & $168 > 153$ & 1{,}213 & LJCR, 2009 \\
$C(30,18,3)$ & 9 & 10 & $(29,17,2)$ & 5 & $180 > 162$ & 747 & LJCR, 2009 \\
$C(33,20,3)$ & 9 & 10 & $(32,19,2)$ & 5 & $198 > 180$ & 3{,}545 & LJCR, 2009 \\
$C(35,21,3)$ & 9 & 10 & $(34,20,2)$ & 5 & $210 > 189$ & 2{,}296 & LJCR, 2009 \\
$C(38,23,3)$ & 9 & 10 & $(37,22,2)$ & 5 & $228 > 207$ & 9{,}375 & LJCR, 2009 \\
$C(40,24,3)$ & 9 & 10 & $(39,23,2)$ & 5 & $240 > 216$ & 6{,}242 & LJCR, 2016 \\
$C(42,25,3)$ & 9 & 10 & $(41,24,2)$ & 5 & $252 > 225$ & 3{,}901 & 777ch, 2018 \\
$C(37,23,3)$ & 8 & 9 & $(36,22,2)$ & 4 & $185 > 184$ & 3 & LJCR, 2009 \\
$C(39,24,3)$ & 8 & 9 & $(38,23,2)$ & 4 & $195 > 192$ & 1 & LJCR, 2009 \\
$C(41,25,3)$ & 8 & 9 & $(40,24,2)$ & 4 & $205 > 200$ & 1 & 777ch, 2020 \\
$C(34,25,4)$ & 8 & 9 & $(33,24,3)$ & 5 & $204 > 200$ & 4 & LJCR, 2009 \\
$C(23,16,4)$ & 10 & 11 & $(22,15,3)$ & 6 & $161 > 160$ & 1 & JCD article, 1996 \\
$C(25,17,4)$ & 10 & 11 & $(24,16,3)$ & 6 & $175 > 170$ & 1 & Kinski and Genivaldo Pereira, 2008 \\
$C(29,20,4)$ & 10 & 11 & $(28,19,3)$ & 6 & $203 > 200$ & 1 & Alessandro Jurcovich, 2008 \\
$C(32,22,4)$ & 10 & 11 & $(31,21,3)$ & 6 & $224 > 220$ & 3 & Kinski and Genivaldo Pereira, 2008 \\
$C(35,24,4)$ & 10 & 11 & $(34,23,3)$ & 6 & $245 > 240$ & 1 & LJCR, 2009 \\
$C(36,25,4)$ & 10 & 11 & $(35,24,3)$ & 6 & $252 > 250$ & 4 & LJCR, 2009 \\
$C(22,17,5)$ & 10 & 11 & $(21,16,4)$ & 7 & $176 > 170$ & 1 & Vinicius Alves, 2008 \\
$C(30,23,5)$ & 10 & 11 & $(29,22,4)$ & 7 & $240 > 230$ & 1 & LJCR, 2009 \\
$C(31,24,5)$ & 10 & 11 & $(30,23,4)$ & 7 & $248 > 240$ & 6 & LJCR, 2009 \\
 
\bottomrule
\end{tabular}
\end{table}

\begin{proof}[Proof of Theorem~\ref{thm:main}]
For each cell, condition (ii) of Theorem~\ref{thm:reduction} is the inequality in
Table~\ref{tab:main}, and conditions (i), (iii) and (iv) follow from the unsatisfiability of an
explicit set of propositional formulas. For (i) the formula states that a doubly lex
$(v-1,k-1,t-1)$-covering with $r-1$ blocks exists, which by Lemma~\ref{lem:lexleast} is
satisfiable exactly when some such covering exists. For (iii) it is $\Flista$, with Theorem~\ref{thm:list}, and for the three cells with $r=4$
Lemma~\ref{lem:four} gives (i) and (iii) independently. For (iv) there is one formula $\Fclase(c)$ per
class, except in the $\nAgrupadas$ cells with $t=3$, $r=5$ and $v\ge\vAgrupadasMin$, where there is one $\Fclase(I)$ per group of up to $\tamGrupo$ consecutive classes. Each of these
$\nFormulasCod$ formulas was written by two encoders (Section~\ref{sec:trust}), and all
$\nFormulasTotal$ were refuted with proofs accepted by a verified checker. The formulas, the DRAT proofs, the class lists, the
logs and the versions of the tools are in the data deposit (see Data availability), and
Theorem~\ref{thm:reduction} gives $C(v,k,t)\ge L+1$.
\end{proof}

For $C(18,11,3)$ and $C(20,12,3)$ the earlier deposit~\cite{zenodo2cells} gives a second proof,
which excludes each of the possible degree sequences by a certified formula.

% ---------------------------------------------------------------------------------------------
\section{What has to be trusted}\label{sec:trust}

Each formula is refuted by kissat~\cite{kissat2024}, which writes a DRAT proof. drat-trim~\cite{drattrim2014}
checks it and converts it to LRAT~\cite{lrat2017}, and cake\_lpr~\cite{cakelpr2021,cakelpr2023},
whose correctness is proved in HOL4 down to its machine code, checks the LRAT proof against the
formula. So the solver and drat-trim need not be trusted, and neither do the enumerators of
links, because a missing class leaves $\Flista$ satisfiable. What remains is the following.
\begin{itemize}
  \item \emph{The argument of Section~\ref{sec:method}.}
  \item \emph{That the formulas say what step (i), Definitions~\ref{def:list} and~\ref{def:ext},
        and the merged formulas $\Fclase(I)$ say.} Two
        encoders were written separately and share only the numbering of the incidence variables
        and the class files. They differ in the encodings of cardinality (sequential
        counter~\cite{sinz2005} against totalizer~\cite{totalizer2003}), of coverage and of lex
        order, and both give unsatisfiable formulas for every case. On small instances the models
        of each encoding were compared exhaustively with the solutions computed without any
        formula.
  \item \emph{That the blocking clauses of $\Flista$ exclude only images of listed classes.}
        Blocking any other matrix could hide a link without making the formula satisfiable. A
        separate program reads every clause of $\Flista$ whose variables are all incidence
        variables. It requires each such clause to have exactly one literal on every incidence
        variable, so that it excludes exactly one matrix, apart from the first clauses of the lex
        chains of one encoder, which it recognises. It then checks that the excluded matrices are
        exactly the union of the sets $\Img^\circ(c_i)$, computed from the class file alone. This
        holds for all \nCeldas{} lists and both encoders.
  \item \emph{The bookkeeping:} a checker confirms for each cell and encoder that the certified
        formulas cover (i), (iii) and every class of the list in (iv), with the parameters of the
        cell, and regenerates every formula byte for byte from its encoder.
  \item \emph{The unverified parts around cake\_lpr:} the verified part is machine code produced
        by the CakeML compiler, and its small C wrapper for input and output, the C compiler and
        linker that build the executable, and the operating system are outside the proof.
\end{itemize}
Positive controls test the encodings at full size. For each cell, the extension formula with
$L+1$ blocks, written for the link of a point of minimum degree in the known covering, is
satisfiable in both encodings, and each model is a covering. Removing one class from a list made
$\Flista$ satisfiable, with a model isomorphic to the removed class, in every run where we tried
it ($\nListasControl$ of the \nCeldas{} lists). These controls are evidence against an
encoding error, not a proof that there is none. The encodings, the tests and the cost of the
computation are described in the long version~\cite{longversion}.

% ---------------------------------------------------------------------------------------------
\section{Scope}\label{sec:scope}

Condition (ii) is rarely what fails. Of the $\nHuecoUno$ cells with a gap of one in the frozen La
Jolla table~\cite{ljcr2026} (version~1.2), with $L$ the lower bound and $r=C(v-1,k-1,t-1)$ as
recorded there when that value is exact, $\nAplica$ satisfy $(r+1)v>kL$. In $\nNoExacto$ more the value of $r$ is not recorded as exact, and in $\nFalla$ the inequality fails because the recorded lower bound lies well above the Sch\"onheim value. In the $\nAplica$ cells the inequality alone settles
nothing: (i), (iii) and (iv) have to be established cell by cell.

Two quantities matter in practice: the number of links and the fraction $r/L$ of the columns that
a link fixes. Among the gap-one cells with $t=3$ and $r=5$ that satisfy (ii)
there are $\nTtresRcinco$, and these are the twelve cells of Table~\ref{tab:main} with $L=9$.
Among the $\nRseisNota$ gap-one cells with $t\in\{3,4\}$ and $r=6$, where $r/L$ lies between
$\rLRseisMin$ and $\rLRseisMax$, canonical augmentation lists the links of $\nRseisEnum$ within ten minutes
each on one thread, with $\rseisClasesMin$ to $\rseisClasesMax$ classes, and $\nRseisTabla$ of these cells are in
Table~\ref{tab:main}. With many classes the cost moves to the formulas: $\Flista$ excludes up to
$r!$ images of each class, and there is one extension formula per class or group.\frasePendientes{} For $t=2$ the reduction is unlikely to help in its present form: all
$\nTdos$ applicable gap-one cells have $r/L<0.3$, so one link leaves most of the matrix free. At
the other end, when $kL=vr$ every point has degree $r$, as in Krug's proof~\cite{krug2026}, so every link, not only
one, is a covering with $r$ blocks.

\section*{Use of AI tools}

From August to October 2026 the author used Claude (Opus 5.5 and Fable 5.1, Anthropic) and GPT-6
Astra (OpenAI) to write the encoders, the link enumerators, the checker and the scripts, to run and
log the computations, to review the mathematics, the code and the writing, and to draft this text
and its long version. The models also took part in finding the arguments, among them the
classification of Lemma~\ref{lem:four}. The author directed the work, checked every proof in
Sections~\ref{sec:prelim} to~\ref{sec:example} step by step and the numbers against the run
records, and is responsible for the content. The values rest on the arguments of
Section~\ref{sec:method} and on cake\_lpr checking the proofs of formulas produced by two
independently written encoders, which the models also helped to write. Section~\ref{sec:trust}
lists what has to be trusted.

\section*{Data availability}

The code (both encoders, the two link enumerators, the checker and the tests), the class lists, the
logs of every formula with the SHA-256 of each CNF, DRAT and LRAT file, the explicit coverings with
$L+1$ blocks and the logs of the controls are deposited on Zenodo,
\zdoi{10.5281/zenodo.23434567}, a DOI that resolves to the latest version of the record. The LRAT
proofs are not deposited and are regenerated from the DRAT proofs with drat-trim. The DRAT proofs
and CNF files of the first \nCeldasLarga{} cells ($\depositoGib$\,GiB) are in four records linked
from the main one, \zdoi{10.5281/zenodo.23434574}, \zdoi{10.5281/zenodo.23434576},
\zdoi{10.5281/zenodo.23434578} and \zdoi{10.5281/zenodo.23441756}, and those of the other
\nCeldasNuevas{} ($\gibNuevas$\,GiB) are in version~2.0 of the main record. The proofs of the degree split of $C(18,11,3)$ and
$C(20,12,3)$ are in~\cite{zenodo2cells}. All these records are released under the CC BY 4.0
licence.

\bibliographystyle{plain}
\bibliography{refs}

\end{document}